\documentclass[11pt,a4paper]{article}
\usepackage{authblk}
\usepackage{amsmath,amssymb,amsfonts,bm,amsthm}

\DeclareFontFamily{U}{mathx}{}
\DeclareFontShape{U}{mathx}{m}{n}{ <-> mathx10 }{}
\DeclareSymbolFont{mathx}{U}{mathx}{m}{n}
\DeclareFontSubstitution{U}{mathx}{m}{n}

\DeclareMathAccent{\widecheck}{0}{mathx}{"71}
\usepackage{mathrsfs}
\usepackage{mathalfa}
\usepackage{graphicx}
\usepackage{color}
\usepackage{lscape}
\usepackage{extarrows}
\usepackage{enumitem}
\usepackage{empheq}
\usepackage{listings}
\usepackage{graphicx}
\usepackage[dvipsnames]{xcolor}
\usepackage{chngcntr}
\usepackage{hyperref}
\usepackage{appendix}
\usepackage{tikz}
\usepackage{comment}
\usepackage{braket}
\usepackage{bbm}
\usepackage{pifont}
\usepackage{physics}
\usepackage[english]{babel}
\usepackage[skins,breakable]{tcolorbox}

\usepackage{shuffle}
\usepackage{tocloft}
\usepackage{cite}

\hypersetup{
	colorlinks=true, 
	linkcolor=violet,  
	urlcolor=blue,    
	anchorcolor=green, 
	citecolor=violet   
}
\newcommand{\T}{\mbox{\tiny T}}

\newcommand{\id}{\mathbbm{1}}

\usepackage[inline,marginal,final]{showlabels}

\theoremstyle{plain}
\newtheorem{theorem}{Theorem}
\newtheorem{result}{Result}
\newtheorem{proposition}{Proposition}
\newtheorem{corollary}{Corollary}
\newtheorem{lemma}{Lemma}

\theoremstyle{definition}
\newtheorem{definition}{Definition}

\theoremstyle{remark}
\newtheorem{remark}{Remark}
\newtheorem{example}{Example}

\newcommand{\RN}[1]{%
	\textup{\uppercase\expandafter{\romannumeral#1}}%
}

\title{Interplay between Standard Quantum Detailed Balance and Thermodynamically Consistent Entropy Production}
\author[1]{Xin-Hai Tong\thanks{\href{mailto:xinhai@iis.u-tokyo.ac.jp}{xinhai@iis.u-tokyo.ac.jp}}}
\affil[1]{Department of Physics, The University of Tokyo, 5-1-5 Kashiwanoha, Kashiwa-shi, Chiba 277-8574, Japan}
\author[2,3]{Kohei Yoshimura}
\affil[2]{Nonequilibrium Quantum Statistical Mechanics RIKEN Hakubi Research Team, Pioneering Research Institute (PRI), RIKEN, 2-1 Hirosawa, Wako, Saitama 351-0198, Japan}
\affil[3]{Universal Biology Institute, The University of Tokyo, 7-3-1 Hongo, Bunkyo-ku, Tokyo 113-0033, Japan}
\author[4]{Tan Van Vu}
\affil[4]{Center for Gravitational Physics and Quantum Information, Yukawa Institute for Theoretical Physics, Kyoto University, Kitashirakawa Oiwakecho, Sakyo-ku, Kyoto 606-8502, Japan}
\author[5]{Naruo Ohga}
\affil[5]{Department of Physics, Graduate School of Science, The University of Tokyo, 7-3-1 Hongo, Bunkyo-ku, Tokyo 113-0033, Japan}
\begin{document} 
	\maketitle

\begin{abstract}
We demonstrate that, for a quantum Markovian semigroup on a finite-dimensional Hilbert space, if it satisfies the standard quantum detailed balance condition, its generator admits a special representation that yields a vanishing entropy production rate. Conversely, if the generator admits a special representation adhering to the condition of thermodynamic consistency and leading to a vanishing entropy production rate, then the corresponding quantum Markovian semigroup must satisfy the standard quantum detailed balance condition. In this context, we adopt the definition of entropy production rate that is motivated by the physics literature and standard for thermodynamically consistent Lindbladians. 
\end{abstract}

\section{Introduction}
The rapid development of quantum platforms \cite{McArdle2020,Blais2021,Sieberer2025} has recently brought significant attention to the nonequilibrium statistical mechanics of open quantum systems \cite{Breuer2002}. To tackle the intrinsically complicated interplay between the system and environmental degrees of freedom, it is common in many scenarios to adopt the Markovian approximation for the dynamics of the system of interest \cite{Gorini1976,Lindblad1976,Seifert2012}. After sufficiently long relaxation times, the reduced density operator of the system typically approaches an invariant (i.e., time-independent) state. A fundamental task is then to determine whether this asymptotic state is a true thermal equilibrium state. If not, it is classified as a non-equilibrium steady state.

It is well known that in classical Markovian dynamics \cite{Seifert2012}, the detailed balance condition serves as a fundamental criterion for distinguishing equilibrium states from non-equilibrium steady states. When the detailed balance condition holds, the entropy production rate (EPR) vanishes, which is consistent with the physical interpretation of the EPR as a quantitative measure of the system's deviation from equilibrium. However, parallel results in the quantum regime require further detailed investigation. This work is devoted to rigorously establishing this exact interplay between the standard quantum detailed balance condition and the vanishing EPR in the quantum regime. \label{reply_intro}

To motivate our approach, we first briefly review the standard results in classical Markovian dynamics for finite state spaces. Let $\Omega$ be a finite set representing the state space. We denote by $\mathcal{A} = \mathbb{C}^{\Omega}$ the commutative algebra of complex-valued functions (observables) on $\Omega$. Let $p_{t}\colon \Omega\to \mathbb{R}_{\geq 0}$ be the probability distribution at time $t$, satisfying the normalization condition $\sum_{i\in \Omega}p_{t}(i)=1$. For distinct $i,j\in \Omega$, let $W(i,j)$ denote the transition rate from state $j$ to $i$. The master equation describing the time evolution of the probability distribution reads
\begin{equation}\label{eq:master_eq}
	\frac{\mathrm{d} }{\mathrm{d} t}p_{t}(i)=\sum_{j\in \Omega:j\neq i}\qty[W(i,j)p_{t}(j)-W(j,i)p_{t}(i)]=\sum_{j\in \Omega}L(i,j)p_{t}(j).
\end{equation}
Here, $L$ is the generator of the dynamics acting on probability distributions, with matrix elements defined by $L(i,j) = W(i,j)$ for $i\neq j$ and $L(i,i) = -\sum_{k\neq i}W(k,i)$. 
The invariant distribution (or steady-state) $\pi$ is a probability distribution satisfying $L\pi = 0$. We say the system satisfies classical detailed balance if
\begin{equation}\label{CDB}
	W(i,j)\pi(j)=W(j,i)\pi(i), \quad \forall i,j\in \Omega.
\end{equation}
For a fixed invariant distribution $\pi$, we define a sesquilinear form on the algebra $\mathcal{A}$ by $(f,g)_{\pi}\coloneq \sum_{i\in \Omega}\overline{f(i)}g(i)\pi(i)$. Strictly speaking, if the support of $\pi$ is not the entire set $\Omega$, this defines a pre-scalar product rather than an inner product. 
Let $L^{\ast}$ denote the adjoint of the generator $L$ with respect to the standard duality pairing $\langle f,g\rangle\coloneq \sum_{i\in \Omega}f(i)g(i)$. In quantum terminology, this adjoint corresponds to the generator in the Heisenberg picture acting on observables in $\mathcal{A}$. Evidently, one can reformulate \eqref{CDB} as 
\begin{equation}\label{classical_inn}
	(f,L^{\ast} g)_{\pi}=(L^{\ast} f,g)_{\pi}, \quad \forall f,g \in \mathcal{A},
\end{equation}
i.e., the classical detailed balance is rigorously characterized by the symmetry (self-adjointness) of the generator $L^{\ast}$ with respect to this pre-scalar product. 
For an invariant distribution $\pi$, the entropy production rate $\sigma$ is well-studied in the physics community and is given by \cite{VandenBroeck2015,Schnakenberg1976}:
\begin{equation}\label{classical_epr}
	\sigma=\frac{1}{2}\sum_{i,j\in \Omega}\qty[W(i,j)\pi(j)-W(j,i)\pi(i)]\ln \frac{W(i,j)\pi(j)}{W(j,i)\pi(i)},
\end{equation}
with the convention $0\ln (0/0)=0$. Comparing \eqref{classical_epr} and \eqref{CDB}, one notes that the EPR vanishes if and only if the classical detailed balance condition holds.

While the relationship between detailed balance and zero EPR is well-understood in the classical setting, the situation becomes significantly more intricate in the quantum context. The challenges in investigating this relationship are twofold.

(\emph{i}) Unlike the classical case, formulating a quantum counterpart of detailed balance directly in terms of the dynamical generator (analogous to \eqref{CDB}) is non-trivial. The standard mathematical approach involves introducing a sesquilinear form over the algebra of bounded operators on a Hilbert space, analogous to \eqref{classical_inn}, and defining detailed balance via the symmetry of the generator with respect to this form. Within this framework, there exist multiple non-equivalent definitions of quantum detailed balance, arising from different choices of the sesquilinear form and the specific relation imposed on the generator and its dual. Historically, this topic has been the subject of extensive discussion over the past five decades. See, e.g., Agarwal \cite{Agarwal1973}, Alicki \cite{Alicki1976}, Frigerio--Gorini--Kossakowski--Verri \cite{Kossakowski1978}, Majewski \cite{Majewski1983,Majewski1984}, Derezi\'nski--Fr\"uboes \cite{Derezinski2006}, and Fagnola--Umanit$\grave{\text{a}}$ \cite{Fagnola2010,Fagnola2015}. \label{reply_history_list}

(\emph{ii}) Second, the definition of EPR for quantum Markovian dynamics is also not unique. A celebrated definition due to Spohn \cite{Spohn1978} identifies EPR with the relative entropy decay. However, if the invariant state is not of the Gibbs form, this definition does not directly correspond to the sum of entropy changes in the system and environment, rendering its thermodynamic interpretation ambiguous. To obtain an EPR with a clear physical interpretation, specific setups have been considered. For example, Breuer \cite{Breuer2003} relates EPR to quantum jumps within the framework of stochastic processes. In \cite{VanVu2023,Horowitz2013,Manzano2018}, extra constraints are imposed on the jump operators to ensure a \textit{thermodynamically consistent} generator, leading to an EPR expression with a transparent physical meaning. Abstract definitions of EPR within the operator algebraic framework have also been discussed in \cite{Jaksic2001,Fagnola2015}.

In this work, we study the relationship between quantum detailed balance and zero EPR by bridging these mathematical and physical perspectives. Specifically, for the former, we adopt the notion of \textit{Standard Quantum Detailed Balance} (SQDB, see Definition \ref{defi_sqdb}). For the latter, we consider the quantum EPR (Definition \ref{def_epr}) defined for thermodynamically consistent generators. Both concepts are well-established and widely adopted in the mathematics and physics communities. Our main contributions, detailed in Section \ref{sec_main_theorem}, are summarized as follows:

\begin{result}[Informal version of Theorem \ref{theorem_sqdb_implies_epr}]
	\textbf{SQDB $\implies$ Zero EPR.} The SQDB condition ensures the existence of a special representation (Definition \ref{def_special_rep}) of the dynamical generator that yields a vanishing EPR.
\end{result}

\begin{result}[Informal version of Theorem \ref{theorem_epr_implies_sqdb}]
	\textbf{Zero EPR $\implies$ SQDB.} Conversely, any thermodynamically consistent dynamical generator with a vanishing EPR must satisfy the SQDB condition.
\end{result}

\noindent Motivated by these results, we further investigate the structural constraints on the space of representations imposed by conditions related to EPR and SQDB. These findings are presented in Section \ref{sec_structure}. We also note that recent related works have explored connections between other variants of detailed balance and entropy production \cite{Benoist2025,Fagnola2015}.

\section{Preliminaries and Notations}
This section is devoted to reviewing the necessary preliminaries, fixing the notation, and collecting standard results for the reader's convenience. Let $\mathcal{H}$ be a finite-dimensional Hilbert space with $\operatorname{dim}\mathcal{H}\eqcolon D$ and $\mathcal{B}(\mathcal{H})$ denote the algebra of all bounded linear operators on $\mathcal{H}$. \label{reply_finite}
We are interested in the quantum Markov semigroup (QMS), defined as a one-parameter family $\mathcal{T}=(\mathcal{T}_{t})_{t\geq 0}$ of normal, completely positive, and unital maps on $\mathcal{B}(\mathcal{H})$ that constitutes a norm-continuous semigroup. The predual semigroup of $\mathcal{T}$ is denoted by $\mathcal{T}_{\ast}=(\mathcal{T}_{\ast t})_{t\geq 0}$. We assume that $\mathcal{T}$ admits a faithful normal invariant state, represented by the density operator $\rho$. Denote by $\id$ the identity in $\mathcal{B}(\mathcal{H})$.
A central result in the theory of QMS is the GKSL representation of the generator, which is a special case of Theorem 30.16 in \cite{Parthasarathy2012}: \label{reply_theorem_rep}

\begin{theorem}\label{theorem_rep}
Let $\mathcal{L}$ be the generator of a norm-continuous QMS on $\mathcal{B}(\mathcal{H})$. There exists a self-adjoint operator $H$ and a finite sequence $(L_\ell)_{\ell \ge 1}$ of elements in $\mathcal{B}(\mathcal{H})$ such that:
\begin{itemize}
	\item[(i)] $\mathrm{Tr}( L_\ell) = 0$ for each $\ell \ge 1$,
	\item[(ii)] The operators in $\{L_{\ell}\}_{\ell\in \mathsf{k}}$ are linearly independent,
	\item[(iii)] the following formula holds for some $H=H ^{\ast}\in \mathcal{B}(\mathcal{H})$:
	\begin{equation}\label{lind_rep}
		\mathcal{L}(x) = i[H, x] - \frac{1}{2} \sum_{\ell \ge 1} (L_\ell^* L_\ell x - 2 L_\ell^* x L_\ell + x L_\ell^* L_\ell).
	\end{equation}
\end{itemize}
If $H'$, $(L'_{\ell})_{\ell \ge 1}$ is another family of bounded operators in $\mathcal{B}(\mathcal{H})$ with $H'$ self-adjoint and the sequence $(L'_\ell)_{\ell \ge 1}$ is finite then the conditions (i)--(iii) are fulfilled with $H$, $(L_\ell)_{\ell \ge 1}$ replaced by $H'$, $(L'_\ell)_{\ell \ge 1}$ respectively if and only if the lengths of the sequences $(L_\ell)_{\ell \ge 1}$, $(L'_{\ell})_{\ell \ge 1}$ are equal and for some scalar $c \in \mathbb{R}$ and a unitary matrix $(u_{\ell j})_{\ell j}$ we have
\begin{equation}\label{}
	\begin{aligned}[b]
		H' = H + c\id, \qquad L'_\ell = \sum_{j\geq 1} u_{\ell j} L_j.
	\end{aligned}
\end{equation}
\end{theorem}
\begin{remark}
It can be checked that the following shift transformation to jump operators and the Hamiltonian also links different representations for the identical generator:
\begin{equation}\label{rep_shift}
	\begin{aligned}[b]
		L_{\ell}'=L_{\ell}+a_{\ell}\id ,\qquad H'=H+\frac{1}{2i}\sum_{\ell\geq 1}\qty(\bar{a}_{\ell}L_{\ell}-a_{\ell}L^{\ast}_{\ell})+b\id,
	\end{aligned}
\end{equation}
for some complex sequence $(a_{\ell})_{\ell\geq 1}\subset \mathbb{C}$.
\end{remark}

In physics, the operators $H$ and $(L_{\ell})_{\ell\geq 1}$ are referred to as the \emph{Hamiltonian} and \emph{jump operators}, respectively.
\begin{definition}\label{def_special_rep}
	Equation \eqref{lind_rep} is said to constitute a \textbf{representation} of the Lindbladian $\mathcal{L}$ in the Heisenberg picture by means of $H$ and $(L_{\ell})_{\ell\geq 1}$ if condition (\emph{iii}) in Theorem \ref{theorem_rep} holds. A representation is called \textbf{special} if conditions (\emph{i}) and (\emph{ii}) are also satisfied.
\end{definition}
\noindent A \emph{special representation} is characterized as a representation consisting of traceless and linearly independent jump operators, a setup relevant to many physical applications.

Since the invariant state $\rho$ is assumed to be faithful, the sesquilinear form 
\begin{equation}\label{inner_quantum}
	(x,y)_{\rho} \coloneq \operatorname{Tr}(\rho^{1/2}x^{\ast} \rho^{1/2} y), \quad x,y \in \mathcal{B}(\mathcal{H})
\end{equation}
constitutes an inner product on $\mathcal{B}(\mathcal{H})$. Note that we reserve the standard notation $(\cdot ,\cdot)$ for inner product over the underlying Hilbert space $\mathcal{H}$. For each QMS $\mathcal{T}=(\mathcal{T}_t)_{t\ge 0}$, the inner product \eqref{inner_quantum} uniquely defines a semigroup $\mathcal{T}'=(\mathcal{T}'_t)_{t\ge 0}$ satisfying the relation
\begin{equation}\label{dual_map}
	(x,\mathcal{T}_{t}(y))_{\rho}=(\mathcal{T}'_{t}(x),y)_{\rho} ,\quad \forall t\in \mathbb{R}_{\geq 0}.
\end{equation}
We refer to $\mathcal{T}'$ as the \emph{dual semigroup} of $\mathcal{T}$ with respect to the invariant state $\rho$. The existence and uniqueness of $\mathcal{T}'$ are guaranteed by Proposition 1 and Theorem 1 in \cite{Fagnola2010}. Let $\mathcal{L}'$ be the generator of $\mathcal{T}'$, then from \eqref{dual_map} one deduces 
\begin{equation}\label{dual_generator}
	\begin{aligned}[b]
	(x,\mathcal{L}(y))_{\rho}=(\mathcal{L}'(x),y)_{\rho},
	\end{aligned}
\end{equation} \label{reply_dual_L}
which serves as the foundation for the following definition of quantum detailed balance, commonly referred to as the \emph{standard quantum detailed balance}  \cite{Fagnola2010,Fagnola2015,Derezinski2006}. As noted in \cite{Fagnola2015}, the original concept dates back to a private communication between Alicki and Majewski.

\begin{definition}\label{defi_sqdb}
	The QMS $\mathcal{T}$ generated by $\mathcal{L}$ satisfies the standard quantum detailed balance (SQDB) condition if there exists a self-adjoint operator $K$ on $\mathcal{H}$ such that $\mathcal{L}(x) - \mathcal{L}'(x) = 2i[K, x]$ for all $x \in \mathcal{B}(\mathcal{H})$, with $\mathcal{L}'$ being defined in \eqref{dual_generator}.
\end{definition}
\begin{remark}
	In some literature (e.g., \cite{Fagnola2010,Fagnola2015}), the duality is formally defined via the bilinear pairing $(x,y)_{0} \coloneqq \operatorname{Tr}(\rho^{1/2}x \rho^{1/2} y)$, rather than the sesquilinear inner product $(x,y)_{\rho}$ used in \eqref{inner_quantum}. Given a faithful state $\rho$, the duality relations $(x, \mathcal{L}(y))_{0}=(\mathcal{L}'(x),y)_{0}$ and $(x, \mathcal{L}(y))_{\rho}=(\mathcal{L}'(x),y)_{\rho}$
	are equivalent for any $\ast$-preserving operator $\mathcal{L}$ (i.e., $\mathcal{L}(x)^{\ast}=\mathcal{L}(x^{\ast})$). Since the generator of a QMS is inherently $\ast$-preserving, these two formulations determine the unique dual operator $\mathcal{L}'$ identically in this context.
\end{remark}
\begin{remark}
Even if $\mathcal{H}$ is infinite-dimensional, the boundedness of $K$ is implied by the self-adjointness of $K$ and the fact that the derivation $x \mapsto [K, x]$ is defined on the entire algebra $\mathcal{B}(\mathcal{H})$.
\end{remark}
The following result, due to Franco Fagnola and Veronica Umanit$\grave{\text{a}}$ \cite{Fagnola2010}, provides a characterization of SQDB:
\begin{theorem}\label{theorem_sqdb_chara}
	The QMS $\mathcal{T}$ satisfies the SQDB if and only if for all special representations of the generator $\mathcal{L}$ by means of operators $H, (L_k)_{k\geq 1}$ there exists a symmetric unitary matrix $(u_{k\ell})_{k,\ell\geq 1}$ such that, for all $k \ge 1$,
	\begin{equation}\label{eq1_theorem_sqdb_chara}
		\begin{aligned}[b]
			\rho^{1/2} L_k^{\ast} = \sum_{\ell} u_{k\ell} L_{\ell} \rho^{1/2}.
		\end{aligned}
	\end{equation}
\end{theorem}
\begin{remark}
	The condition ``for all special representation'' is equivalent to ``for a single special representation'', since all special representations are related by a unitary transformation.
\end{remark}

We now introduce the notation for the quantum EPR. Let $\mathsf{k}$ be a finite set representing the indices of all jump operators in some representation of a QMS generator. In the context of quantum thermodynamics, it is essential to consider representations where the jump operators appear in pairs. We encode this structure via an involution on the index set, denoted by $\ast: \mathsf{k}\rightarrow \mathsf{k}$.
We distinguish between indices based on their behavior under this involution: for $\ell\in \mathsf{k}$, if $\ell^{\ast}=\ell$, we term the index a \textit{singlet} and if $\ell^{\ast}\neq \ell$, the pair $[\ell, \ell^{\ast}]$ is termed a \textit{doublet}. It is important to emphasize that the involution  is an internal structure of the index set $\mathsf{k}$ and is, by construction, representation independent.
With these preparations, we define the quantum EPR in the following. \label{reply_def_epr}
\begin{definition}\label{def_epr}
Let $\rho$ be an invariant state over a finite-dimensional Hilbert space $\mathcal{H}$, and let us fix a spectral decomposition $\sum_{j}\rho_{j}\ket{e_{j}}\bra{e_{j}}$ of it. In the presence of degenerate eigenvalues, we arbitrarily choose and fix one such orthonormal eigenbasis $\{e_j\}_{j}$. Consider a representation of the generator of a QMS over $\mathcal{B}(\mathcal{H})$ with finitely many jump operators (i.e., $|\mathsf{k}|<\infty$), characterized by $H$ and $(L_{\ell})_{\ell\in \mathsf{k}}$. Denote $w_{ij}^{k}\coloneq |(e_{i},L_{k}e_{j})|^{2}$ and the quantum EPR is defined as follows
\begin{align}\label{quantum_epr}
	\sigma=\frac{1}{2}\sum_{k\in \mathsf{k}}\sum_{i,j}\qty(\rho_{i}w^{k}_{ji}-\rho_{j}w^{k^{\ast}}_{ij})\ln\frac{\rho_{i}w^{k}_{ji}}{\rho_{j}w^{k^{\ast}}_{ij}}
\end{align}
with the convention $0\ln (0/0)=0$.
\end{definition} 
\begin{remark}
	Let $\widetilde{\mathsf{k}}$ denote the collection of orbits in $\mathsf{k}$ under the involution $k \mapsto k^*$ \footnote{Strictly speaking, $\widetilde{\mathsf{k}}$ represents the quotient space $\mathsf{k}/\mathbb{Z}_2$ under the group action generated by the involution $k \mapsto k^{\ast}$.}, categorized as doublets (orbits of size 2) or singlets (orbits of size 1). One may reformulate the summation in \eqref{quantum_epr} as
	\begin{equation}\label{epr_reduced}
		\sigma = \sum_{k\in \widetilde{\mathsf{k}}}\theta_{k}\sum_{i,j}\left(\rho_{i}w^{k}_{ji}-\rho_{j}w^{k^{\ast}}_{ij}\right)\ln\frac{\rho_{i}w^{k}_{ji}}{\rho_{j}w^{k^{\ast}}_{ij}},
	\end{equation}
	where the index $k\in \widetilde{\mathsf{k}}$ in the outer summation represents an arbitrary element selected from the corresponding doublet or singlet. The weight factor $\theta_{k}$ is defined as $1$ for doublets and $1/2$ for singlets.
\end{remark}
\noindent It has been repeatedly discussed in the physics literature \cite{VanVu2023,Horowitz2013,Manzano2018} that the quantum EPR defined in \eqref{quantum_epr} admits a transparent thermodynamic interpretation, provided that the jump operators satisfy a so-called local detailed balance relation. Specifically, there exist positive constants $c_{k}$ for each $k\in \mathsf{k}$ such that\begin{equation}\label{tc}L_{k^{\ast}} = c_{k}L_{k}^{\ast}.\end{equation}Note that the involution property $k^{**}=k$ imposes the necessary consistency constraint $c_{k}c_{k^{\ast}}=1$. In this context, the quantity $-2\ln c_{k}\eqcolon s_{k}$ is identified as the entropy change in the environment associated with the jump operator $L_{k}$. We emphasize that the definition of the EPR in \eqref{quantum_epr} is fundamentally quantum-motivated and basis-independent. Indeed, under the condition \eqref{tc}, one can verify that the EPR reduces to $\sigma=\sum_{k\in \mathsf{k}}s_{k}\operatorname{Tr}(L^{\ast}_{k}L_{k}\rho)$, which exactly captures the entropy production in the environment. Furthermore, the entropy change of the system is given by the time derivative of its von Neumann entropy, which strictly vanishes for an invariant state. The definition \eqref{quantum_epr} can also be generalized to any transient state $\rho(t)$ by simply replacing the constant eigenvalues $\{\rho_{j}\}$ with the time-dependent ones $\{\rho_{j}(t)\}$. In this general setting, under the condition \eqref{tc}, the quantum EPR precisely identifies the total entropy production: the sum of the entropy changes in both the system and the environment. For a more detailed and related discussion, we refer the readers to Appendix E of \cite{VanVu2023}. \label{reply_epr_page}

\begin{definition}\label{def_tc}
	A representation $H,(L_{k})_{k\in \mathsf{k}}$ of a QMS generator is said to be \textbf{thermodynamically consistent} (TC) with respect to the family $(c_{k})_{k\in \mathsf{k}}$, if the finite index set $\mathsf{k}$ is equipped with an involution $k \mapsto k^*$ and the relation \eqref{tc} holds for all $k\in \mathsf{k}$.
\end{definition}

\section{Interplay between SQDB and Zero EPR}\label{sec_main_theorem}
In this section, we investigate the correspondence between the SQDB condition (Definition \ref{defi_sqdb}) and the quantum EPR introduced in the preceding section. Consistent with the setup in Definition \ref{def_epr}, we continue to restrict our attention to finite-dimensional Hilbert spaces and generators represented by finitely many jump operators. Our main results are summarized in Theorems \ref{theorem_sqdb_implies_epr} and \ref{theorem_epr_implies_sqdb} below. Their proofs are established by invoking Theorem \ref{theorem_sqdb_chara} together with several auxiliary lemmas and propositions. \label{reply_main_theorem_3}

\begin{theorem}\label{theorem_sqdb_implies_epr}
Let $\mathcal{T}$ be a QMS satisfying SQDB  and  $\mathsf{k}$ be a finite set equipped with an involution. If the generator $\mathcal{L}$ admits a special representation indexed by $\mathsf{k}$, then there exists a special representation indexed by the same involutive index set $\mathsf{k}$, such that the quantum EPR defined in \eqref{quantum_epr} vanishes.
\end{theorem}
\begin{remark}
The trivial case where there are only singlets in $\mathsf{k}$ is also included in the involution structure.
\end{remark}
\begin{proof}
We enter the special representation $H,(L_{\ell})_{\ell\in \mathsf{k}}$ described in the assumption and by Theorem \ref{theorem_sqdb_chara} there exists a symmetric unitary matrix $(u_{k \ell})_{k,\ell\in \mathsf{k}}$ such that
\begin{equation}\label{result_the_cha}
	\begin{aligned}[b]
		\rho^{1/2} L_k^{\ast} = \sum_{\ell} u_{k\ell} L_{\ell} \rho^{1/2}.
	\end{aligned}
\end{equation}
From Theorem \ref{theorem_rep}, a unitary matrix $(v_{ij})_{k,\ell\in \mathsf{k}}$ links this representation $(L_{\ell})_{\ell \in \mathsf{k}}$ to another one, denoted by $(J_{\ell})_{\ell\in \mathsf{k}}$ via
\begin{equation}\label{find_ano_rep}
	\begin{aligned}[b]
		L_{k}=\sum_{a}v_{ka}J_{a} \quad\text{and thus}\quad L_{k}^{\ast}=\sum_{a}\bar{v}_{ka}J_{a}^{\ast}.
	\end{aligned}
\end{equation}
Putting \eqref{find_ano_rep} into \eqref{result_the_cha}, we obtain
\begin{equation}\label{find_ano_rep_2}
	\begin{aligned}[b]
		\sum_{a}\bar{v}_{ka}\rho^{1/2}J_{a}^{\ast}=\sum_{a}(uv)_{ka}J_{a}\rho^{1/2} \quad\text{hence}\quad \rho^{1/2}J_{k}^{\ast}=\sum_{a}(\bar{v}^{-1}uv)_{ka}J_{a}\rho^{1/2}.
	\end{aligned}
\end{equation}
Provided that we require
\begin{equation}\label{required_v}
	\begin{aligned}[b]
		(\bar{v}^{-1}uv)_{ka}=\delta_{k^{\ast} a}
	\end{aligned}
\end{equation}
with $\delta$ denoting the Kronecker delta, we have $\rho^{1/2}J_{k}^{\ast}=J_{k^{\ast}}\rho^{1/2}$. This relation implies $(e_{i},\rho^{1/2}J_{k}^{\ast}e_{j})=(e_{i},J_{k^{\ast}}\rho^{1/2}e_{j})$ for any eigenvectors of $\rho$. Therefore, one deduces $\rho_{i}^{1/2}(e_{i},J_{k}^{\ast}e_{j})=(e_{i},J_{k^{\ast}}e_{j})\rho_{j}^{1/2}$ and this relation
implies $		w_{ji}^{k}\rho_{i}=w^{k^{\ast}}_{ij}\rho_{j}$, which clearly leads to a vanishing EPR [cf.\,(\ref{quantum_epr})].

It therefore suffices to show that there exists a unitary matrix $v$ satisfying \eqref{required_v}. First we consider a specific ordering of the index set $\mathsf{k}$ such that doublets precede singlets. Formally, the sequence of jump operators is arranged so that $(2i-1)^{\ast}=2i$ for $1 \le i \le n$ (where $n$ is the number of doublets) and $i^{\ast}=i$ for $i > 2n$. Under this ordering, the involution matrix $\Delta = (\delta_{k^{\ast}a})_{k,a\in \mathsf{k}}$ takes the block-diagonal form:
\begin{equation}\label{J_matrix}
	\mathcal{J} \coloneq \left(\bigoplus_{i=1}^{n} \sigma_{x}\right) \oplus \left(\bigoplus_{j=1}^{m} (1)\right),
	\quad \text{with } \sigma_{x} \coloneq \begin{pmatrix} 0 & 1 \\ 1 & 0 \end{pmatrix},
\end{equation}
where $m$ denotes the number of singlets. Then generally, for an arbitrary ordering there always exists a permutation matrix $P$ such that the involution matrix is given by $\Delta = P \mathcal{J} P^{\T}$. 
Now, consider the symmetric unitary matrix $u$ arising from \eqref{result_the_cha}. Lemma \ref{lemma_autonue_takagi}, which we state below, guarantees the existence of a unitary matrix $S$ such that $S^{\T} u S = \mathcal{J}$. By identifying $S = v P$, we establish the existence of the required unitary transformation $v = S P^{\T}$, which is unitary as the product of two unitary matrices (noting that $\overline{v}^{-1} = v^{\T}$).
\end{proof}

\begin{lemma}\label{lemma_autonue_takagi}
	Let $n$ and $m$ be non-negative integers, and let $U$ be a $(2n+m) \times (2n+m)$ complex matrix that is both symmetric  and unitary.
	Define the matrix
	$
	\sigma_{x} = \begin{pmatrix} 0 & 1 \\ 1 & 0 \end{pmatrix}$.
	Then, there exists a $(2n+m) \times (2n+m)$ unitary matrix $S$ such that
	\begin{equation} \label{eq:takagi_decomposition}
		S^{\T}US = \left(\bigoplus_{i=1}^{n} \sigma_{x}\right) \oplus \left(\bigoplus_{j=1}^{m} (1)\right) = \mathcal{J}.
	\end{equation}
\end{lemma}
\begin{proof}
	From Autonne--Takagi factorization \cite{Autonne1915,Takagi1924}, for any symmetric complex matrix $U$, there exists a unitary matrix $V$ such that $V^{\T}UV=\Sigma$ with $\Sigma$ being the diagonal matrix filled by the singular values of $U$. In our case, the matrix $U$ is also unitary, implying all of its singular values are $1$. Therefore, we have 
	\begin{equation}\label{}
		\begin{aligned}[b]
			V^{\T}UV=I_{2n+m}.
		\end{aligned}
	\end{equation}  
	with $I_{2n+m}$ being $(2n+m)\times (2n+m)$ identity matrix.
	Next, we define $W_{2}\coloneq 2^{-1/2}\begin{pmatrix} 1 & 1 \\ i & -i \end{pmatrix}$ and $I_{2}\coloneq \begin{pmatrix} 1 & 0 \\ 0 & 1 \end{pmatrix}$. A direct calculation shows that $W^{\T}I_{2}W=\sigma_{x}$ and the unitarity $W^{\ast}W=I_{2}$. Therefore, we construct $Q = (\bigoplus_{i=1}^{n} W_{2}) \oplus (\bigoplus_{j=1}^{m} (1))$ such that $Q^{\T}I_{2n+m}Q=\mathcal{J}$ and have 
	\begin{equation}\label{}
		\begin{aligned}[b]
			Q^{\T}V^{\T}U  V Q= \left(\bigoplus_{i=1}^{n} \sigma_{x}\right) \oplus \left(\bigoplus_{j=1}^{m} (1)\right).
		\end{aligned}
	\end{equation}
	By denoting $S=VQ$ and checking its unitarity we complete the proof.
\end{proof}
\begin{remark}\label{remark_unique_non}
	The matrix $S$ is not unique here. In fact, the matrix $V'=VO$ with any real orthogonal matrix $O$ will also give the same Autonne--Takagi factorization. On the other hand, for each sector, the $W$ matrix can always be replaced by $W_{2}(\theta) \coloneq 2^{-1/2}\begin{pmatrix} e^{i\theta} &e^{-i\theta} \\ ie^{i\theta} & -ie^{-i\theta} \end{pmatrix}$ for any $\theta\in \mathbb{R}$ such that we still have $W^{\T}(\theta)I_{2}W(\theta)=\sigma_{x}$ and the unitarity $W(\theta)^{\ast}W(\theta)=I_{2}$.
\end{remark}

\label{reply_proof_theorem_epr_implies_sqdb}
\begin{theorem}\label{theorem_epr_implies_sqdb}
A QMS $\mathcal{T}$ generated by $\mathcal{L}$ satisfies SQDB if it admits a special representation that satisfies the TC condition (Definition \ref{def_tc}) and whose corresponding quantum EPR (defined in Definition \ref{def_epr}) vanishes.
\end{theorem}
\begin{proof}
We enter the representation mentioned in the assumption and from the expression of EPR \eqref{quantum_epr} we have
\begin{equation}\label{sigma_0}
	\begin{aligned}[b]
		\sigma=0 \quad\text{and thus}\quad \rho_{i}w_{ji}^{k}=\rho_{j}w_{ij}^{k^{\ast}} 
	\end{aligned}
\end{equation}
for all $i,j\in \{1,2,...,\dim\mathcal{H}\}$.
Because this representation is TC, we have
\begin{equation}\label{tc_app}
	\begin{aligned}[b]
		L_{k^{\ast}}=c_{k}L_{k}^{\ast}\quad\text{hence}\quad (e_{i},L_{k^{\ast}}e_{j})=c_{k}(e_{i},L_{k}^{\ast} e_{j}),\quad \forall k\in \mathsf{k}.
	\end{aligned}
\end{equation}
Equation \eqref{sigma_0} implies $\rho_{i}|(e_{j},L_{k}e_{i})|^{2}=\rho_{j}|(e_{i},L_{k^{\ast}}e_{j})|^{2}$ and further $\rho_{i}|(e_{i},L_{k}^{\ast}e_{j})|^{2}=\rho_{j}|(e_{i},L_{k^{\ast}}e_{j})|^{2}$ for all $i,j\in \{1,2,...,\dim\mathcal{H}\}$. With \eqref{tc_app}, one deduces that 
\begin{equation}\label{ck2_rhoij}
	\begin{aligned}[b]
		\rho_{i}|(e_{i},L_{k}^{\ast}e_{j})|^{2}=c_{k}^{2}\rho_{j}\qty|(e_{i},L_{k}^{\ast}e_{j})|^{2} \quad\text{or equivalently}\quad 	(c_{k}^{2}-\rho_{i}/\rho_{j})|(e_{i},L_{k}^{\ast}e_{j})|^{2}=0
	\end{aligned}
\end{equation}
Since the eigenvalues of $\rho$ are strictly positive, \eqref{ck2_rhoij} implies that either $(e_{i},L_{k}^{\ast}e_{j})=0$ or $c_{k}=\sqrt{\rho_{i}/\rho_{j}}$. In either case, the identity $c_{k}(e_{i},L_{k}^{\ast}e_{j})=\sqrt{\rho_{i}/\rho_{j}}(e_{i},L_{k}^{\ast}e_{j})$ holds. By noticing the TC condition, we obtain the element-wise equality
\begin{equation}\label{elements_wise}
	\begin{aligned}[b]
		(e_{i},L_{k^{\ast}}e_{j})=(e_{i},\rho^{1/2}L_{k}^{\ast}\rho^{-1/2}e_{j})
	\end{aligned}
\end{equation}
for all $i,j$. Consequently, the following operator equality holds true:
\begin{equation}\label{L_k_ast_eq}
	\begin{aligned}[b]
		L_{k^{\ast}}=\rho^{1/2}L_{k}^{\ast}\rho^{-1/2}.
	\end{aligned}
\end{equation}
Note that the existence and boundedness of $\rho^{-1/2}$ follow from the finite dimensionality of $\mathcal{H}$ and faithfulness of $\rho$. Equation~\eqref{L_k_ast_eq} can be reformulated as
\begin{equation}\label{}
	\begin{aligned}[b]
		\rho^{1/2}L_{k}^{\ast}=L_{k^{\ast}}\rho^{1/2}=\sum_{\ell}\delta_{k^{\ast}\ell}L_{\ell}\rho^{1/2}.
	\end{aligned}
\end{equation}
Note that the matrix $(\delta_{k^{\ast}\ell})_{k,\ell\in \mathsf{k}}$, which equals to $P\mathcal{J}P^{\T}$ from the discussion below \eqref{J_matrix}, is both unitary and symmetric. It suffices to find a single special representation for which \eqref{eq1_theorem_sqdb_chara} holds. Since due to the unitary equivalence, Eq.\eqref{eq1_theorem_sqdb_chara} is then valid for all special representations. Consequently, by Theorem \ref{theorem_sqdb_chara}, we conclude that $\mathcal{L}$ satisfies the SQDB condition.
\end{proof}

The main results, Theorems \ref{theorem_sqdb_implies_epr} and \ref{theorem_epr_implies_sqdb}, can be intuitively summarized as follows. 
\begin{equation*}\label{}
	\begin{aligned}[b]
		\exists \text{ TC zero-EPR special representation} \implies &\text{ SQDB } \implies \exists \text{ zero-EPR special representation}
	\end{aligned}
\end{equation*}
We use the following example to emphasize that the SQDB condition is insufficient to imply the existence of a zero-EPR special representation satisfying the TC condition. This fact is quite natural since SQDB does not encode thermodynamical information such as $(c_{k})_{k\in \mathsf{k}}$.
\begin{example}
We consider a two-level system, where the generator of the corresponding QMS can be represented by only one jump operator. Specifically, we set $H=\omega \sigma_{z}=\begin{pmatrix} \omega & 0 \\ 0 & -\omega \end{pmatrix}$ and $L=\sqrt{\nu}\sigma_{+}+\sqrt{1-\nu}\sigma_{-}=\begin{pmatrix} 0 & \sqrt{\nu} \\ \sqrt{1-\nu} & 0 \end{pmatrix}$ with $\omega\in \mathbb{R}_{>0}$ and $\nu\in (0,1)\backslash\{1/2\}$. Here, $\sigma_z$ is the Pauli-$z$ matrix, and $\sigma_{\pm} = (\sigma_x \pm i \sigma_y)/2$ are the raising and lowering operators. One can verify that $(H,L)$ here constitute a special representation and $\rho=\begin{pmatrix} \nu & 0 \\ 0 & 1-\nu \end{pmatrix}$ gives the faithful invariant state, since $-i[H,\rho]+L\rho L^{\ast}-\{L^{\ast}L,\rho\}/2=0$. One can also check that $\rho^{1/2}L^{\ast}=L\rho^{1/2}=\begin{pmatrix} 0 & \sqrt{\nu(1-\nu)} \\ \sqrt{\nu(1-\nu)} & 0 \end{pmatrix}$, which implies the QMS here satisfies SQDB from Theorem \ref{theorem_sqdb_chara} with the unitary matrix $(1)$.  However, in this case the generator will never have a special representation satisfying the TC condition. To see this, we first note that the TC condition imposes that in any special representation, there exists a  nonzero scalar $c$ such that $\widetilde{L}=c\widetilde{L}^{\ast}$. Here $\widetilde{L}=e^{i\theta}L$ with $\theta\in \mathbb{R}$ denotes another special representation. This relation implies $e^{i\theta}\begin{pmatrix} 0 & \sqrt{\nu} \\ \sqrt{1-\nu} & 0 \end{pmatrix}=ce^{-i\theta}\begin{pmatrix} 0 & \sqrt{1-\nu} \\ \sqrt{\nu} & 0 \end{pmatrix}$ and therefore imposing $e^{2i\theta}\sqrt{\nu}=c\sqrt{1-\nu}$, $e^{2i\theta}\sqrt{1-\nu}=c\sqrt{\nu}$. The latter two equations cannot be satisfied by any scalar $c$ unless $\nu = 1/2$.
\end{example}

\section{Structure of Representations Involving EPR and SQDB}\label{sec_structure}
Continuing with the finite-dimensional framework, we now examine the structure of the representation space. Recall from Theorem \ref{theorem_sqdb_implies_epr} that SQDB guarantees the existence of at least one zero-EPR special representation.

Focusing on the jump operators, let $\mathfrak{R}(\mathcal{T})$ be the collection of all families $(L_k)_{k \in \mathsf{k}}$ forming a special representation for the generator of $\mathcal{T}$ (with $H$ fixed). We also define the involution matrix $\Delta = (\delta_{k^{\ast}\ell})_{k,\ell\in \mathsf{k}}$. 
Based on these definitions, the following proposition enables one to find other zero-EPR special representations in $\mathfrak{R}(\mathcal{T})$, starting from the one constructed in Theorem  \ref{theorem_sqdb_implies_epr} when $\mathcal{T}$ satisfies SQDB,
\begin{proposition}\label{prop_zero_EPR_trans}
	Let $\mathsf{k}$ be a finite set equipped with an involution and $(L_{\ell})_{\ell \in \mathsf{k}}$ be a special representation in $ \mathfrak{R}(\mathcal{T})$ such that $\rho^{1/2}L_{k}^{\ast}=L_{k^{\ast}}\rho^{1/2}$ holds for all $k\in \mathsf{k}$. An alternative special representation $(\widetilde{L}_{\ell})_{\ell \in \mathsf{k}}$ of the identical generator leads to vanishing EPR,  if $\widetilde{L}_{k}=\sum_{\ell}v_{kl}L_{\ell}$ for some unitary matrix $v$ satisfying $v^{\T}\Delta v= \Delta$.
\end{proposition}
\begin{proof}
	We first note that the relation $\rho^{1/2}L_{k}^{\ast}=L_{k^{\ast}}\rho^{1/2}$ can be reformulated as 
	\begin{equation}\label{prop1_before_trans}
		\begin{aligned}[b]
			\rho^{1/2}L_{k}^{\ast}=\sum_{\ell}\Delta_{k\ell}L_{\ell}\rho^{1/2}.
		\end{aligned}
	\end{equation}
	Since $v$ is unitary, the condition $v^{\T}\Delta v=\Delta$ is equivalent to $\Delta v =\bar{v} \Delta$. Then from \eqref{prop1_before_trans} we have
    	\begin{equation}\label{}
    	\begin{aligned}[b]
    		\rho^{1/2}\sum_{k}\bar{v}_{ak}L_{k}^{\ast}=\sum_{\ell,k}\bar{v}_{ak}\Delta_{k\ell}L_{\ell}\rho^{1/2}=\sum_{\ell}(\Delta v)_{a\ell}L_{\ell}\rho^{1/2}.
    	\end{aligned}
    \end{equation}
This  immediately yields
 	\begin{equation}\label{prop1_after_trans}
	\begin{aligned}[b]
		\rho^{1/2}\widetilde{L}_{a}=\sum_{b}\Delta_{ab}\widetilde{L}_{b}\rho^{1/2},
	\end{aligned}
\end{equation}
   since $\widetilde{L}_{k}=\sum_{\ell}v_{kl}L_{\ell}$. Clearly, the relation \eqref{prop1_after_trans} implies that the EPR vanishes in the representation $(\widetilde{L}_{\ell})_{\ell\in \mathsf{k}}$. Note that this representation is also special due to the unitarity of $v$, as guaranteed by Theorem \ref{theorem_rep}.
\end{proof}
In light of this result, the following proposition demonstrates how to find any other TC special representations.  
\begin{proposition}\label{prop_LDB_pres_zero_epr}
	Let $(L_{\ell})_{\ell \in \mathsf{k}}$ be a special representation in $\mathfrak{R}(\mathcal{T})$ satisfying TC with respect to $(c_{k})_{k\in \mathsf{k}}$. An alternative special representation $(\widetilde{L}_{\ell})_{\ell \in \mathsf{k}}$ of the identical generator satisfies TC with respect to $(\widetilde{c}_{k})_{k\in \mathsf{k}}$, if and only if $\widetilde{L}_{k}=\sum_{\ell}v_{kl}L_{\ell}$ for some unitary matrix $v$ satisfying $Cv^{\T}\widetilde{C}^{-1}\Delta v=\Delta$ with $C=\operatorname{diag}\{c_{k}\}$ and $\widetilde{C}=\operatorname{diag}\{\widetilde{c}_{k}\}$.
\end{proposition}
\begin{proof}
By assumption,
	\begin{subequations}\label{}
		\begin{align}
			&	L_{k^{\ast}}=c_{k}L_{k}^{\ast} \quad\text{and thus}\quad \sum_{\ell}\Delta_{k\ell}L_{\ell}=c_{k}L^{\ast}_{k}\label{ldb_j_k_1} \\ 
			& \widetilde{L}_{k^{\ast}}=\widetilde{c}_{k}\widetilde{L}_{k}^{\ast} \quad\text{and thus}\quad \sum_{\ell}\Delta_{k\ell}\widetilde{L}_{\ell}=\widetilde{c}_{k}\widetilde{L}^{\ast}_{k}\label{ldb_j_k_2}.
		\end{align}
	\end{subequations}
	From Theorem \ref{theorem_rep}, these two families of jump operators are linked by a unitary matrix $v$, i.e., $\widetilde{L}_{k}=\sum_{\ell}v_{k\ell}L_{\ell}$. Putting this relation into \eqref{ldb_j_k_2} we have
	\begin{equation}\label{}
		\begin{aligned}[b]
			\sum_{a}(\Delta v)_{ka}L_{a}=\widetilde{c}_{k}\sum_{a}\bar{v}_{ka}L_{a}^{\ast}=\widetilde{c}_{k}\sum_{a}c_{a}^{-1}\bar{v}_{ka}\sum_{\ell}\Delta_{a\ell}L_{\ell}=\sum_{a}\qty(\sum_{b}\widetilde{c}_{k}\bar{v}_{kb}c_{b}^{-1}\Delta_{ba})L_{a}.
		\end{aligned}
	\end{equation}
	Here, for the second equal sign we used \eqref{ldb_j_k_1}. Then by the linear independence of the jump operators, we have $\Delta v= \widetilde{C} \bar{v} C^{-1}\Delta$.
	Noting that $v$ is unitary, we obtain $Cv^{\T}\widetilde{C}^{-1}\Delta v=\Delta$. Conversely, if $\Delta v= \widetilde{C} \bar{v} C^{-1}\Delta$ holds, one obtain \eqref{ldb_j_k_2} by using $\widetilde{L}_{k}=\sum_{\ell}v_{k\ell}L_{\ell}$.
	 Therefore, we finish the proof.
\end{proof}
\begin{corollary}
	Let $(L_{\ell})_{\ell \in \mathsf{k}}$ be a TC (with respect to $(c_{k})_{k\in \mathsf{k}}$) special representation in $ \mathfrak{R}(\mathcal{T})$ with vanishing EPR. An alternative TC (with respect to $(\widetilde{c}_{k})_{k\in \mathsf{k}}$) special representation $(\widetilde{L}_{\ell})_{\ell \in \mathsf{k}}$ of the identical generator leads to vanishing EPR,  if and only if $\widetilde{L}_{k}=\sum_{\ell}v_{kl}L_{\ell}$ for some unitary matrix $v$ satisfying both $v^{\T}\Delta v= \Delta$ and $Cv^{\T}\widetilde{C}^{-1}\Delta v=\Delta$. Here, we denote $C=\operatorname{diag}\{c_{k}\}$ and $\widetilde{C}=\operatorname{diag}\{\widetilde{c}_{k}\}$.
\end{corollary}
\begin{proof}
	First, recall from the discussion in Theorem \ref{theorem_epr_implies_sqdb} that for a TC representation $(L_{\ell})_{\ell\in \mathsf{k}}$, the condition of vanishing EPR is equivalent to
	\begin{equation}\label{tc_epr_zero}
		\rho^{1/2}L_{k}^{\ast}=L_{k^{\ast}}\rho^{1/2}=\sum_{\ell}\Delta_{k\ell}L_{\ell}\rho^{1/2}.
	\end{equation}
	
	For the ``if'' part, The condition $v^{\T}\Delta v=\Delta$ ensures vanishing EPR by Proposition \ref{prop_zero_EPR_trans}, while $Cv^{\T}\widetilde{C}^{-1}\Delta v=\Delta$ ensures that the representation $(\widetilde{L}_{\ell})_{\ell\in \mathsf{k}}$ is TC with respect to $(\widetilde{c}_{k})_{k\in \mathsf{k}}$, as established in Proposition \ref{prop_LDB_pres_zero_epr}. The representation remains special due to the unitarity of $v$ and Theorem \ref{theorem_rep}.
	For the ``only if'' part. Assume that the TC representation $(\widetilde{L}_{\ell})_{\ell\in \mathsf{k}}$ yields a vanishing EPR. This implies the relation:
	\begin{equation}\label{cor1_1}
		\rho^{1/2}\widetilde{L}_{k}^{\ast}=\widetilde{L}_{k^{\ast}}\rho^{1/2}=\sum_{\ell}\Delta_{k\ell}\widetilde{L}_{\ell}\rho^{1/2}.
	\end{equation}
	By Theorem \ref{theorem_rep}, there exists a unitary matrix $v$ such that $\widetilde{L}_{k}=\sum_{\ell}v_{k\ell}L_{\ell}$. Substituting this into \eqref{cor1_1} yields $\rho^{1/2}\sum_{\ell}\overline{v}_{k\ell}L_{\ell}^{\ast} = \sum_{\ell, a}\Delta_{k\ell}v_{\ell a}L_{a}\rho^{1/2}$. Using \eqref{tc_epr_zero} to expand the left-hand side, we obtain:
	\begin{equation}
		\sum_{\ell, a}\overline{v}_{k\ell}\Delta_{\ell a}L_{a}\rho^{1/2}=\sum_{\ell, a}\Delta_{k\ell}v_{\ell a}L_{a}\rho^{1/2} \quad\text{hence}\quad \sum_{\ell, a}\overline{v}_{k\ell}\Delta_{\ell a}L_{a}=\sum_{\ell, a}\Delta_{k\ell}v_{\ell a}L_{a},
	\end{equation}
	where we have used the faithfulness of $\rho$. By the linear independence of the jump operators $(L_{\ell})_{\ell\in \mathsf{k}}$ in special representation, the matrix equality $\overline{v}\Delta = \Delta v$ holds is equivalent to $v^{\T}\Delta v= \Delta$ given the unitarity of $v$. The second condition $Cv^{\T}\widetilde{C}^{-1}\Delta v=\Delta$ is guaranteed by Proposition \ref{prop_LDB_pres_zero_epr} since both representations are TC and special.
\end{proof}

\begin{remark}
	It turns out that the two multisets of coefficients $\{c_{k}\}_{k\in \mathsf{k}}$ and $\{\widetilde{c}_{k}\}_{k\in \mathsf{k}}$ must coincide. To see this, note that the involution matrix $\Delta$ is invertible (note that $\Delta= P\mathcal{J}P^{\T}$ from discussion below \eqref{J_matrix}). From the relations $v^{\T}\Delta v= \Delta$ and $Cv^{\T}\widetilde{C}^{-1}\Delta v=\Delta$, we deduce that $Cv^{\T}\widetilde{C}^{-1} = v^{\T}$. This implies $C = v^{\T}\widetilde{C}(v^{\T})^{-1}$ and therefore the diagonal matrices $C$ and $\widetilde{C}$ are unitarily equivalent. Consequently, their spectra (i.e., the multisets of diagonal elements) coincide.
\end{remark}
In summary, the unitarity of $v$ is required to link other special representations. Then one needs to impose $v^{\T}\Delta v=\Delta$ to preserve zero EPR while imposing $Cv^{\T}\widetilde{C}^{-1}\Delta v=\Delta$ to preserve TC (from $\{c_{k}\}$ to $\{\widetilde{c}_{k}\}$). 

While this work primarily focuses on special representations of the QMS generator, they are not always the natural starting point in physical setups. To address this, we show that any given representation can be transformed into a special one. 
Consider a general representation where the jump operators $(L_{k})_{k\in \mathsf{k}}$ are not necessarily linearly independent. We first introduce the shifted operators $\widetilde{L}_{k} = L_{k}-\operatorname{Tr}(L_{k}) \id/D$, ensuring that all $\widetilde{L}_{k}$ are traceless. Note that to preserve the generator, the Hamiltonian must be adjusted according to \eqref{rep_shift}. Next, the summation $\sum_{k\in \mathsf{k}}\widetilde{L}_{k}^{\ast}\widetilde{L}_{k}$  can always be rewritten as a standard form $\sum_{\ell\in \mathsf{k}'}J_{\ell}^{\ast}J_{\ell}$ where the operators $(J_{\ell})_{\ell\in \mathsf{k}'}$ are linearly independent (the number of jump operators may reduce, i.e., $|\mathsf{k}'|\leq |\mathsf{k}|$). Crucially, this transformation preserves the dissipative part (see Definition \ref{def_dissi_cp_part} later) of the generator, namely $-\sum_{k\in \mathsf{k} } (\widetilde{L}_k^* \widetilde{L}_k x - 2 \widetilde{L}_k^* x \widetilde{L}_k + x \widetilde{L}_k^* \widetilde{L}_k)/2 = -\sum_{\ell\in \mathsf{k}'}(J_{\ell}^{\ast}J_{\ell}x -2J_{\ell}^{\ast}xJ_{\ell}+xJ_{\ell}^{\ast}J_{\ell})/2$. Consequently, together with the adjusted Hamiltonian, the family $(J_{\ell})_{\ell\in \mathsf{k}'}$ yields a representation of the dynamical generator identical to that of $(L_{k})_{k\in \mathsf{k}}$. Since each $J_{\ell}$ is a linear combination of the traceless operators $(\widetilde{L}_{k})_{k\in \mathsf{k}}$, it remains traceless, i.e., $\operatorname{Tr}(J_{\ell})=0$. This satisfies condition (i) of a special representation.

We now turn to the construction of a special representation that preserves TC. Starting from a linearly independent TC representation $(L_{k})_{k\in \mathsf{k}}$ with coefficients $(c_{k})_{k\in \mathsf{k}}$, we first shift jump operators $\widetilde{L}_{k} = L_{k}-\operatorname{Tr}(L_{k}) \id/D$ to ensure tracelessness.
It is straightforward to verify that this shift preserves the TC relation:
\begin{equation}
	c_{k}\widetilde{L}_{k}^{\ast} = c_{k}\qty[L^{\ast}_{k} -\overline{\operatorname{Tr}(L_{k})} \id/D] = L_{k^{\ast}} -\operatorname{Tr}(L_{k^{\ast}}) \id/D  = \widetilde{L}_{k^{\ast}}.
\end{equation}
However, whether the family $(\widetilde{L}_{k})_{k\in \mathsf{k}}$ remains linearly independent depends on the algebraic structure of the original operators. Suppose first that $\id \notin \operatorname{span}(\{L_{k}\}_{k\in \mathsf{k}})$. Then by Lemma \ref{lemma_shift_linear}, the shifted family $(\widetilde{L}_{k})_{k\in \mathsf{k}}$ remains linearly independent. Since these operators are traceless, the resulting representation is both special and TC. On the other hand, if $\id \in \operatorname{span}(\{L_{k}\}_{k\in \mathsf{k}})$, i.e., $\id=\sum_{k}b_{k}L_{k}$, it follows that $\sum_{k}b_{k}\operatorname{Tr}(L_{k})/D = \operatorname{Tr}(\id)/D = 1$. Consequently, Lemma \ref{lemma_shift_linear} implies that $\{\widetilde{L}_{k}\}_{k\in \mathsf{k}}$ is linearly dependent. In this scenario, it is impossible to construct a special TC representation using this shift method.
\begin{lemma}\label{lemma_shift_linear}
	Let $(L_k)_{k=1}^n$ be a family of linearly independent operators and $(a_k)_{k=1}^n$ be a family of complex numbers. The family $(J_k \coloneq L_k - a_k \id)_{k=1}^n$ is linearly dependent if and only if the identity operator lies in the span of $\{L_k\}_{k=1}^n$ (i.e., $\id = \sum_{k=1}^n b_k L_k$) and the coefficients satisfy the condition $\sum_{k=1}^n a_k b_k = 1$.
\end{lemma}
\begin{proof}
	For the ``if'' part, it is sufficient to find a family of coefficients $\{c_{k}\}_{k=1}^{n}$, not all zero, such that $\sum_{k=1}^{n}c_{k}J_{k}=0$. In this case, we simply choose $c_{k}=b_{k}$ and verify:
	\begin{equation}
		\sum_{k=1}^{n}b_{k}J_{k} = \sum_{k=1}^{n}b_{k}L_{k} - \left(\sum_{k=1}^{n}b_{k}a_{k}\right)\id = \id - 1 \cdot \id = 0.
	\end{equation}
	Since $\sum_{k=1}^n b_k L_k=\id \neq 0$, the coefficients $b_k$ cannot be all zero, which proves linear dependence.
	For the ``only if'' part, suppose that $\{J_k\}_{k=1}^{n}$ is linearly dependent. Then there exists a family of coefficients $\{c_{k}\}_{k=1}^{n}$, not all zero, such that $\sum_{k}c_{k}J_{k}=0$. This relation is equivalent to  $\sum_{k=1}^{n}c_{k}L_{k} = \left(\sum_{k=1}^{n}c_{k}a_{k}\right)\id$.
	One confirms that the scalar factor on the right-hand side must be non-zero. Otherwise, the linear independence of $\{L_{k}\}_{k=1}^{n}$ would imply $c_{k}=0$ for all $k$, contradicting the assumption. Therefore, we can divide by this scalar to obtain $\id= \sum_{k=1}^{n}b_{k}L_{k}$ with $b_{k} \coloneq c_{k}/\sum_{j=1}^{n}c_{j}a_{j}$. The condition $\sum_{k=1}^{n}a_{k}b_{k}=1$ follows immediately.
\end{proof}
Next we consider the implications of SQDB to parts of the QMS generator. Precisely, we give the following definition,
\begin{definition}\label{def_dissi_cp_part}
Let $(L_\ell)_{\ell \in \mathsf{k}}$ be a representation of a QMS generator $\mathcal{L}$. Then the maps $\mathcal{D}$ and $\Phi$ defined by $\mathcal{D}(x) := - \sum_{\ell \in \mathsf{k}} (L_\ell^* L_\ell x - 2 L_\ell^* x L_\ell + x L_\ell^* L_\ell)/2$ and $\Phi(x) := \sum_{\ell \in \mathsf{k}} L_\ell^* x L_\ell$ for all $x \in \mathcal{B}(\mathcal{H})$ are called the dissipative part and the completely positive part of $\mathcal{L}$, respectively.
\end{definition}
\begin{remark}
The dissipative and completely positive parts are invariant under all special representations, as a consequence of the unitary equivalence established in Theorem \ref{theorem_rep}.
\end{remark}

\begin{proposition}\label{prop_dissi_sqdb}
	If a QMS admits a representation (not necessarily special) such that the dissipative part of the generator is symmetric with respect to the inner product defined in \eqref{inner_quantum}, then this QMS satisfies the SQDB condition.
\end{proposition}
\begin{proof}
	Let $\mathcal{L}_{\ast}$ be the generator of the predual semigroup $\mathcal{T}_{\ast}$ and $\mathcal{D}_{\ast}$ be its dissipative part. The relation $\mathcal{L}(x) = i[H, x] + \mathcal{D}(x)$ implies $\mathcal{L}_{\ast}(\rho) = -i[H, \rho] + \mathcal{D}_{\ast}(\rho)$. 
	By the assumption that $\mathcal{D}$ is symmetric with respect to the inner product \eqref{inner_quantum}, we have
	\begin{equation}
		(y, \mathcal{D}(x))_{\rho} = (\mathcal{D}(y), x)_{\rho}, \quad \forall x, y \in \mathcal{B}(\mathcal{H}).
	\end{equation}
	Setting $y = \id$, the left-hand side becomes $ \operatorname{Tr}(\rho^{1/2}\id \rho^{1/2} \mathcal{D}(x)) = \operatorname{Tr}(\rho \mathcal{D}(x)) = \operatorname{Tr}(\mathcal{D}_*(\rho) x)$ and the right-hand side vanishes, since $\mathcal{D}(\id) = 0$ by construction.
	Thus, $\operatorname{Tr}(\mathcal{D}_*(\rho) x) = 0$ for all $x \in \mathcal{B}(\mathcal{H})$, which implies $\mathcal{D}_{\ast}(\rho) = 0$.
	Consequently, the invariance of the state $\rho$ implies
	\begin{equation}\label{com_H_rho}
		0 = \mathcal{L}_{\ast}(\rho) = -i[H,\rho] + \mathcal{D}_{\ast}(\rho) = -i[H,\rho] \quad\text{and thus}\quad [H,\rho] = 0.
	\end{equation}
	We now define the operator $\mathcal{L}' \coloneq -i[H, \cdot] + \mathcal{D}$ and verify that it is indeed the dual of $\mathcal{L}$. The symmetry of the dissipative part $\mathcal{D}$ is guaranteed by the assumption. For the Hamiltonian part, we compute the dual of $\mathcal{H}(\cdot) \coloneq i[H, \cdot]$:
	\begin{equation}
		\begin{aligned}
			(x, i[H, y])_{\rho} &= \operatorname{Tr}(\rho^{1/2} x^{\ast} \rho^{1/2} i[H,y])= i\operatorname{Tr}([\rho^{1/2} x^{\ast} \rho^{1/2} ,H]y)=i\operatorname{Tr}(\rho^{1/2}[ x^{\ast}  ,H]\rho^{1/2}y)\\&=\operatorname{Tr}(\rho^{1/2}(-i[H  ,x])^{\ast}\rho^{1/2}y)=(-i[H,x],y)_{\rho}.
		\end{aligned}
	\end{equation}
	Here, we have used $[H,\rho^{1/2}]=0$ since $[H,\rho]=0$.
	Thus, the dual of $\mathcal{L}$ is $\mathcal{L}' = -i[H, \cdot] + \mathcal{D}$. Finally, we observe that $\mathcal{L}(x) - \mathcal{L}'(x) = 2i[H, x]$, which fits the definition of SQDB with $K=H$. The uniqueness of the dual operator ensures that $\mathcal{L}'$ is the only candidate, see discussion below \eqref{dual_map}.
\end{proof}
\begin{remark}
	This result provides an algebraic criterion to verify SQDB without explicitly invoking the quantum EPR. While this statement may be considered standard folklore among experts, we include a rigorous derivation here for completeness, as an explicit formulation appears to be absent from the standard literature.
\end{remark}
\begin{remark}
	As a corollary, the proof above shows that if a representation yields a symmetric dissipative part with respect to the inner product \eqref{inner_quantum}, then the associated Hamiltonian must commute with the invariant state $\rho$.
\end{remark}
Proposition \ref{prop_dissi_sqdb} elucidates the interplay between the dissipative part of a QMS generator and the SQDB condition. Analogous results also hold for the completely positive part. As established in \cite{Fagnola2010}, a QMS $\mathcal{T}$ satisfies SQDB if and only if the completely positive part of its generator $\mathcal{L}$ is symmetric with respect to the inner product in special representations.

\section*{Acknowledgements} 
X.-H.T. is particularly grateful to Zongping Gong and Koki Shiraishi for their valuable comments and insights which significantly improved the manuscript. X.-H.T. also thanks Naomichi Hatano for his continuous support and for organizing the POS-RST 2025 conference, which facilitated stimulating discussions. We acknowledge Ryuji Takagi, Takashi Mori, Ryusuke Hamazaki, and Ao Yuan for fruitful discussions.
X.-H.T. was supported by the FoPM, WINGS Program, the University of Tokyo. K.Y. was supported by the Special Postdoctoral Researchers Program at RIKEN, JST ERATO Grant No.~JPMJER2302, and JSPS KAKENHI Grant No.~24H00834. T.V.V. was supported by JSPS KAKENHI Grant No.~JP23K13032, No.~JP26K00022, and No.~JP26H02015.  N.O. was supported by JSPS KAKENHI Grant No.~23KJ0732.

\section*{Conflict of Interest} 
The authors declare that they have no conflict of interest.

\section*{Data Availability}
Data sharing is not applicable to this article as no datasets were generated or analyzed during the current study.

\nocite{Bibtexkey}
\bibliographystyle{unsrt}
\bibliography{ref}

\end{document}